\documentclass[11pt,letterpaper]{article}

\usepackage[T1]{fontenc}
\usepackage[utf8]{inputenc}
\usepackage[margin=1in]{geometry}
\usepackage{amsmath,amssymb,amsthm}
\usepackage{mathtools}
\usepackage{bm}
\usepackage{booktabs}
\usepackage{multirow}
\usepackage{array}
\usepackage{xcolor}
\usepackage{graphicx}
\usepackage{caption}
\usepackage{subcaption}
\usepackage{enumitem}
\usepackage{mdframed}
\usepackage{parskip}
\usepackage{float}
\usepackage{placeins}
\usepackage{afterpage}

\usepackage{tikz}
\usetikzlibrary{shapes.geometric,arrows.meta,positioning,calc}
\usepackage[numbers,sort&compress]{natbib}
\usepackage[colorlinks=true,
            linkcolor=blue!60!black,
            citecolor=blue!60!black,
            urlcolor=blue!60!black,
            pdftitle={Current-Sheet Formation in Electron Magnetohydrodynamics with Split Fractional Dissipation},
            pdfauthor={Hu, Peng, Yang}]{hyperref}

\numberwithin{equation}{section}
\numberwithin{figure}{section}
\numberwithin{table}{section}

\theoremstyle{plain}
\newtheorem{theorem}{Theorem}[section]

\newtheorem{proposition}[theorem]{Proposition}

\theoremstyle{definition}

\newcommand{\T}{\mathbb{T}}
\newcommand{\R}{\mathbb{R}}
\newcommand{\Z}{\mathbb{Z}}
\newcommand{\bB}{\bm{B}}
\newcommand{\bJ}{\bm{J}}
\newcommand{\nperp}{\nabla^{\perp}}
\newcommand{\curl}{\nabla\times}
\newcommand{\Div}{\nabla\cdot}

\newcommand{\Lal}{\Lambda^{\alpha}}
\newcommand{\Lbe}{\Lambda^{\beta}}
\newcommand{\norm}[2]{\left\|#1\right\|_{#2}}
\newcommand{\Linf}{L^{\infty}}
\newcommand{\Ltwo}{L^{2}}
\newcommand{\Hs}{H^{s}}
\newcommand{\ie}{\textit{i.e.}}

\newcommand{\sect}[1]{Section~\ref{#1}}

\begin{document}

\title{\textbf{Current-Sheet Formation in Electron Magnetohydrodynamics with Split Fractional Dissipation}}

\author{Ruimeng Hu\thanks{Department of Mathematics and
  Department of Statistics and Applied Probability,
  University of California, Santa Barbara, CA 93106, USA.
  \texttt{rhu@ucsb.edu}}
\and
Qirui Peng\thanks{Department of Mathematics,
  University of California, Santa Barbara, CA 93106, USA.
  \texttt{qpeng9@ucsb.edu}}
\and
Xu Yang\thanks{Department of Mathematics,
  University of California, Santa Barbara, CA 93106, USA.
  \texttt{xy6@ucsb.edu}}}

\date{}
\maketitle
\thispagestyle{plain}

\begin{abstract}
Thin current sheets are central small-scale structures in electron magnetohydrodynamics (EMHD), closely associated with energy dissipation and fast magnetic reconnection at electron scales. We study their formation numerically in a $2\frac{1}{2}$-dimensional EMHD system on a periodic domain with split fractional dissipation, where the magnetic potential and the vertical magnetic component are damped separately. The local theory is governed by a symmetric combined damping balance, but the numerical onset of small-scale growth need not follow this symmetry. A scaling analysis identifies the out-of-plane current as the primary concentration observable, since it is regularized only through the magnetic-potential equation. Using a validated Fourier pseudospectral exponential time-differencing solver with resolution-controlled diagnostics, we find a clear decay/concentration dichotomy. The onset boundary is markedly asymmetric: current-sheet formation appears to be controlled mainly by damping of the magnetic potential, rather than by the combined damping strength. The analyticity strip collapses to the grid scale, the concentration sharpens under grid refinement, and the observed growth is {suggestive of an energy-critical self-similar rate, with a tentative exponent near three}. These experiments indicate that magnetic-potential damping is the apparent binding constraint for current-sheet concentration {in the regime studied}.
\end{abstract}

\bigskip
\noindent\textbf{Keywords:} Electron magnetohydrodynamics; Hall effect; fractional dissipation; current-sheet formation; self-similar scaling.

\medskip
\noindent\textbf{Mathematics Subject Classification:} 76W05, 35Q35, 65M70, 35R11.

\section{Introduction}\label{sec:intro}

Thin current sheets are the central coherent structures of electron magnetohydrodynamics (EMHD), the fluid model of plasma dynamics at scales below the ion inertial length, where the magnetic field is frozen into the electron flow~\citep{GKR1994}. Their formation is closely associated with energy dissipation and fast magnetic reconnection at electron scales, and simulations of electron-MHD turbulence show the current density organizing into sheet-like structures across scales~\citep{BSD1996,BSZD1999}. This paper asks a quantitative question about the gate to this process: when the two magnetic degrees of freedom of the $2\frac{1}{2}$-dimensional model are damped through separate dissipation channels, which channel controls the onset of current-sheet formation? The split-dissipation system itself is a recently introduced model: its local well-posedness was established only in~\cite{Peng2026}. The experiments below show that it already exhibits a distinctive nonlinear phenomenology, including an asymmetric onset boundary, amplitude-dependent thresholds, and {tentative evidence for} self-similar concentration at an energy-critical rate. It therefore provides a minimal setting in which the role of damping structure in small-scale formation can be isolated and quantified.

The mathematical difficulty is set by the Hall term. The parent Hall--magnetohydrodynamics (Hall--MHD) system models plasma dynamics in regimes where the Hall effect becomes relevant at small length scales; relative to classical magnetohydrodynamics, the Hall term carries one additional spatial derivative of the magnetic field, producing a genuinely quasilinear difficulty that dominates both the analysis and the numerics of these equations~\citep{ADFL2011,CDL2014}. EMHD is the fluid-free subsystem obtained by retaining the Hall dynamics while neglecting the ion velocity; it serves as a standard model for small-scale plasma motion and whistler-type phenomena~\citep{BSD1996,BSZD1999}.

We study a $2\frac{1}{2}$-dimensional (2.5D) reduction in which the magnetic field is independent of the vertical coordinate and is written as
\begin{equation}
  \bB \;=\; \curl(a\,\bm{e}_z) + b\,\bm{e}_z
  \;=\; (a_y,\,-a_x,\,b),
  \label{eq:ansatz}
\end{equation}
with $a=a(x,y,t)$ the magnetic potential and $b=b(x,y,t)$ the vertical magnetic component. Allowing distinct fractional dissipation on the two scalar components yields the system
\begin{subequations}\label{eq:emhd}
\begin{align}
  \partial_t a + a_y b_x - a_x b_y &= -\Lal a, \label{eq:emhd-a}\\
  \partial_t b - a_y \Delta a_x + a_x \Delta a_y &= -\Lbe b, \label{eq:emhd-b}
\end{align}
\end{subequations}
on $(0,\infty)\times\T^2$, where $\Lambda^s := (-\Delta)^{s/2}$ is the Fourier multiplier with symbol $|k|^s$, $k\in\Z^2$, and throughout $0<\alpha,\beta<2$. Fractional orders below two interpolate between the undamped and fully resistive regimes and serve as a standard mathematical proxy for weakened or anomalous damping; assigning distinct orders to the two components turns the question of which damping channel binds into a quantitative one.

\paragraph{Related literature.}
A systematic local well-posedness theory for three-dimensional Hall--MHD was developed by~\cite{CDL2014}, with blow-up criteria and small-data results in~\cite{CL2014}; lower-regularity and critical-space theories followed in \citep{DanchinTan2021,DanchinTan2022,LiuTan2021,Dai2021}. The role of weakened, fractional magnetic diffusion was clarified by~\cite{CWW2015}, whose Littlewood--Paley and Besov estimates show that the full Laplacian dissipation can be relaxed to a fractional one. For the fluid-free EMHD subsystem, the Hall nonlinearity carries a derivative falling on the current $\curl\bB$, and the behavior is especially delicate without resistivity: \cite{JeongOh2022} proved ill-posedness near degenerate stationary states, while~\cite{JeongOh2025} established well-posedness around nonzero uniform fields; \cite{Dai2023} studied a resistive $2$D EMHD setting near a steady state. Recently, \cite{DaiBabaei2025} proved local well-posedness of the $2.5$D EMHD system~\eqref{eq:emhd} when {one} of the two equations carries a full Laplacian dissipation.  In~\cite{Peng2026}, the second author established the local-well-posedness under a {split} fractional condition
\begin{equation}
  0<\alpha,\beta<2, \qquad \alpha+\beta>2,
  \label{eq:threshold}
\end{equation}
in which case neither component need carry a full derivative of dissipation: the combined smoothing of the two fractional operators suffices to control the Hall nonlinearity. The mechanism is an exact cancellation between the leading low--high frequency interactions of the two equations, after which the residual terms close in an asymmetric energy in which $a$ is estimated one derivative above $b$. Beyond the deterministic setting, the well-posedness theory has recently been extended to a relaxed EMHD model with random diffusion~\citep{HPY2025} and to the three-dimensional stochastic EMHD system, for which maximal pathwise solutions were constructed~\citep{HPY2026}.

The condition~\eqref{eq:threshold} is a sufficient condition for {local} well-posedness; it closes the known energy estimates but does not, by itself, establish global regularity, nor does it preclude small-scale growth or finite-time blow-up for $\alpha+\beta>2$. Our aim is therefore not to test a global regularity threshold but to ask a sharper, complementary question: when small scales do form, which of the two dissipations controls their onset, and does the symmetric balance $\alpha+\beta$ that governs the analysis also govern the observed numerical onset of current-sheet growth? We address this through a combination of scaling heuristics and resolution-controlled numerical experiments, taking care to distinguish what the computations show (a concentration mechanism and its parametric dependence) from what they cannot settle (the existence of a genuine finite-time singularity). Numerical studies of small-scale and near-singular dynamics are by now a mature tool in the analysis of nonlinear PDEs: tracking complex singularities through spectral data~\citep{SSF1983}, continuation criteria of Beale--Kato--Majda type~\citep{BKM1984}, and high-resolution investigations of near-singular dynamics~\citep{LuoHou2014}. We adopt that toolkit here, with particular attention to separating genuine small-scale structure from under-resolution.

In this work, we combine scaling heuristics with resolution-controlled numerical experiments to study how the split dissipations affect small-scale formation in the $2.5$D EMHD system. The scaling analysis identifies three competing mechanisms: a per-equation criticality, the symmetric sum $\alpha+\beta$ arising from the known cancellation-based local theory, and a possible $\alpha$-controlled mechanism associated with the out-of-plane current $\Delta a$. The numerical experiments point toward the last scenario: the observed current-sheet concentration occurs primarily in $\Delta a$ and appears more sensitive to the dissipation on the magnetic potential than to the sum $\alpha+\beta$. {Power-law fits give a tentative estimate of the concentration rate close to the energy-critical value $q^*=3$.} The scaling analysis is developed in \sect{sec:scaling}; the solver, its validation, and the diagnostic and resolution-control protocol in \sect{sec:method} and \sect{sec:diagnostics}; and the numerical results in \sect{sec:results}.

A structure-preserving forward solver for~\eqref{eq:emhd} based on gradient recovery is developed in the companion work~\citep{GHPY2026}; the present paper is concerned not with method design but with using high-accuracy computation to investigate how the local-theory balance~\eqref{eq:threshold} relates to the observed onset of small-scale growth.

\section{The \texorpdfstring{$2.5$D}{2.5D} EMHD system} \label{sec:structure}

This section introduces the notation and structural identities used throughout the paper: the transport form of the system, the current density, and the magnetic-energy balance.

\subsection{Transport form and the Hall nonlinearity}

Introduce the in-plane velocity $u := \nperp b = (-b_y, b_x)$, which is divergence free, $\Div u = 0$. Since $u\cdot\nabla a = -b_y a_x + b_x a_y = a_y b_x - a_x b_y$, equation~\eqref{eq:emhd-a} is a transport--dissipation equation for the magnetic potential,
\begin{equation} 
\partial_t a + u\cdot\nabla a = -\Lal a, \qquad u = \nperp b . \label{eq:a-transport} 
\end{equation} 
 For the second equation, using that $\Delta$ commutes with $\partial_x,\partial_y$ we have $a_x\Delta a_y - a_y\Delta a_x = a_x(\Delta a)_y - a_y(\Delta a)_x = \{a,\Delta a\}$, where $\{f,g\}:=\nperp f\cdot\nabla g = f_x g_y - f_y g_x$ is the Poisson bracket. Hence
\begin{equation} 
\partial_t b + \nperp a\cdot\nabla(\Delta a) = -\Lbe b . \label{eq:b-transport} 
\end{equation} 
The Hall nonlinearity $\nperp a\cdot\nabla(\Delta a)$ in~\eqref{eq:b-transport} pairs one derivative of $a$ against {three} derivatives of $a$; this is the quasilinear, derivative-losing structure responsible both for the difficulty of the well-posedness theory and for the appearance of $a$ at one higher level of regularity than $b$ in the energy.

\subsection{Current density: the singularity observable}

The current density associated with the ansatz~\eqref{eq:ansatz} is
\begin{equation}
  \bJ = \curl\bB = (b_y,\,-b_x,\,-\Delta a),
  \label{eq:current}
\end{equation}
so that the in-plane current is $\nperp b$ and the out-of-plane current is $-\Delta a$. The scalar
\begin{equation}
  \norm{\bJ}{\Linf}(t) \;\sim\; \max\!\big(\norm{\nabla b}{\Linf},\,
  \norm{\Delta a}{\Linf}\big)
  \label{eq:Jinf}
\end{equation}
is the EMHD analogue of the vorticity maximum in incompressible flow and serves as our primary singularity observable. Tracking $\norm{\nabla b}{\Linf}$ and $\norm{\Delta a}{\Linf}$ separately resolves whether an incipient singularity is of in-plane or out-of-plane (current-sheet) type.

\subsection{Conserved magnetic energy}

Define the magnetic energy
\begin{equation}
  M(t) \;:=\; \tfrac12\int_{\T^2}\!\big(|\nabla a|^2 + b^2\big)\,dx
  \;=\; \tfrac12\int_{\T^2}\! |\bB|^2\,dx .
  \label{eq:energy}
\end{equation}

\begin{proposition}[Magnetic-energy balance]\label{prop:energy}
For smooth solutions of~\eqref{eq:emhd},
\begin{equation}
  \frac{d}{dt}M(t)
  = -\norm{\Lambda^{1+\alpha/2} a}{\Ltwo}^2 - \norm{\Lambda^{\beta/2} b}{\Ltwo}^2
  \;\le\; 0 ,
  \label{eq:energy-balance}
\end{equation}
and in particular $M(t)$ is conserved for the dissipationless system.
\end{proposition}

\begin{proof}
Differentiating~\eqref{eq:energy} and integrating by parts,
\[
  \frac{d}{dt}M = \int \nabla a\cdot\nabla\partial_t a + \int b\,\partial_t b
  = -\int \Delta a\,\partial_t a + \int b\,\partial_t b .
\]
Insert $\partial_t a = -u\cdot\nabla a - \Lal a$ and $\partial_t b = -\nperp a\cdot\nabla\Delta a - \Lbe b$. The dissipative terms give $\int \Delta a\,\Lal a = -\norm{\Lambda^{1+\alpha/2}a}{\Ltwo}^2$ and $-\int b\,\Lbe b = -\norm{\Lambda^{\beta/2}b}{\Ltwo}^2$. It remains to show the nonlinear terms cancel. The transport term contributes $\int \Delta a\,(u\cdot\nabla a)$, while the Hall term contributes $-\int b\,\nperp a\cdot\nabla\Delta a = \int \Delta a\,(\nperp a\cdot\nabla b)$, where we integrated by parts using $\Div\nperp a = 0$. Now $u\cdot\nabla a = \nperp b\cdot\nabla a = -(a_x b_y - a_y b_x)$ and $\nperp a\cdot\nabla b = a_x b_y - a_y b_x$, so the two integrands are $\int\Delta a\,[-(a_xb_y-a_yb_x)]$ and $\int\Delta a\,(a_xb_y-a_yb_x)$, which sum to zero. This proves~\eqref{eq:energy-balance}.
\end{proof}

Proposition~\ref{prop:energy} plays two roles below. Analytically, $M$ is the base coercive quantity on which the higher-order well-posedness energy of~\cite{Peng2026} is built. Numerically, the structural identity behind $\tfrac{d}{dt}M=0$, the exact cancellation of the two nonlinear contributions, provides a stringent discrete check (\sect{subsec:validation}): the time integrator need not conserve $M$ exactly as a discrete invariant, but any implementation that fails to reproduce this nonlinear energy cancellation to round-off has the wrong nonlinear structure and cannot be trusted in the near-critical regime.

\section{Scaling, criticality, and competing predictions}\label{sec:scaling}

This section derives the self-similar scaling of the dissipationless system, identifies the criticality of each fractional dissipation, and extracts three competing threshold predictions that the experiments will distinguish.

\subsection{Self-similar scaling of the dissipationless system}

Consider the dissipationless system
\begin{equation}
  \partial_t a + \nperp b\cdot\nabla a = 0, \qquad
  \partial_t b + \nperp a\cdot\nabla\Delta a = 0 .
  \label{eq:inviscid}
\end{equation}

\begin{proposition}[Scaling family]\label{prop:scaling}
If $(a,b)$ solves~\eqref{eq:inviscid}, then for every $\lambda>0$ and every $q\in\R$ so do the rescaled fields
\begin{equation}
  a_\lambda(x,t) = \lambda^{p}\,a\!\Big(\frac{x}{\lambda},\frac{t}{\lambda^{q}}\Big),
  \qquad
  b_\lambda(x,t) = \lambda^{r}\,b\!\Big(\frac{x}{\lambda},\frac{t}{\lambda^{q}}\Big),
  \qquad
  p = 3-q, \quad r = 2-q .
  \label{eq:scaling}
\end{equation}
\end{proposition}

\begin{proof}
Write $\xi = x/\lambda$, $\sigma = t/\lambda^q$. In~\eqref{eq:emhd-a} with no dissipation, $\partial_t a_\lambda \sim \lambda^{p-q}$ while $\nperp b_\lambda\cdot\nabla a_\lambda \sim \lambda^{(r-1)+(p-1)} = \lambda^{p+r-2}$; equality of exponents forces $r = 2-q$. In the second equation, $\partial_t b_\lambda \sim \lambda^{r-q}$ while $\nperp a_\lambda\cdot\nabla\Delta a_\lambda \sim \lambda^{(p-1)+(p-3)} = \lambda^{2p-4}$; equality forces $r-q = 2p-4$, \ie\ $p = 3-q$ after substituting $r=2-q$.
\end{proof}

Two consequences of~\eqref{eq:scaling} fix the relevant observables. First, the current scales isotropically across its in-plane and out-of-plane parts: $\Delta a_\lambda \sim \lambda^{p-2} = \lambda^{1-q}$ and $\nabla b_\lambda \sim \lambda^{r-1} = \lambda^{1-q}$, so
\begin{equation}
  \norm{\bJ}{\Linf} \sim \lambda^{1-q},
  \label{eq:J-scaling}
\end{equation}
and a focusing collapse (length scale $\lambda\to0$ with $\norm{\bJ}{\Linf}\to \infty$) requires $q>1$. This identifies $\bJ$ as the natural normalizing quantity for dynamic rescaling (\sect{subsec:rescaling}). Second, since $p = r+1$, the two contributions to the magnetic energy scale identically, $M\sim\lambda^{2p} = \lambda^{6-2q}$, consistent with Proposition~\ref{prop:energy}. The exponent $6-2q$ separates three regimes: for $q<3$ the magnetic energy of the collapsing core {vanishes} as $\lambda\to0$ (energy-subcritical collapse), $q=3$ is the energy-critical scaling at which $M$ is exactly scale-invariant, and $q>3$ would require a diverging core energy and is incompatible with the conserved, finite $M$. {The tentative estimate} $q^*\approx 3$ in \sect{sec:results} {is close to} the energy-critical value, a point we return to below.

For later use we denote the temporal power laws implied by a self-similar collapse with exponent $q$. If the length scale $L(t)$ closes at the singular time $t^*$ in the self-similar fashion of \sect{subsec:rescaling}, then $t^*-t\sim L^{q}$, so $L\sim(t^*-t)^{1/q}$; combining with~\eqref{eq:J-scaling} and $\delta\sim L$ for the analyticity-strip width gives
\begin{equation}
  \norm{\bJ}{\Linf}\sim(t^*-t)^{-\gamma_J},\quad
  \delta(t)\sim(t^*-t)^{\nu},\qquad
  \gamma_J=\frac{q-1}{q},\quad \nu=\frac1q,\qquad
  q = 1 + \frac{\gamma_J}{\nu}.
  \label{eq:exponent-relation}
\end{equation}
The individual exponents $\gamma_J$ and $\nu$ depend on the (extrapolated) singular time: over a narrow fitting window, misplacing $t^*$ approximately rescales both fitted exponents by a similar factor, so the ratio $\gamma_J/\nu$, and hence $q$ through~\eqref{eq:exponent-relation}, {may be less sensitive to $t^*$} than either exponent alone{, although over a narrow window it remains only loosely determined}. We use this in \sect{sec:results}.

\subsection{Criticality of the fractional dissipation}

Restoring the dissipation, $\Lal a_\lambda \sim \lambda^{p-\alpha}$ is to be compared with the inviscid balance $\lambda^{p-q}$ in~\eqref{eq:emhd-a}: the $a$-dissipation is scale \emph{marginal} when $\alpha = q$, scale \emph{subcritical} (negligible at small scales, unable to arrest collapse) when $\alpha<q$, and \emph{supercritical} (regularizing) when $\alpha>q$. The identical computation for~\eqref{eq:emhd-b} gives marginality at $\beta=q$. Thus, at a self-similar collapse with exponent $q$, dimensional analysis applied to each equation {in isolation} suggests that blow-up can proceed only if both $\alpha<q$ and $\beta<q$, a split-dependent, corner-shaped condition in the $(\alpha,\beta)$ square.

\subsection{Three readings, and the role of the out-of-plane current}

The corner prediction {differs from} the analytical threshold~\eqref{eq:threshold}, which closes on $\alpha+\beta>2$ rather than on $\max(\alpha,\beta)$. One resolution lies in the cancellation identified in~\citep{Peng2026}: the leading low--high frequency interactions of the two equations cancel exactly, so the two fractional dissipations are not independently tasked with regularizing their own equation. Instead they supply a \emph{joint} smoothing budget, of order $\tfrac{\alpha}{2}$ on $a$ and $\tfrac{\beta}{2}$ on $b$ beyond the energy level, set against the single one-derivative loss of the Hall nonlinearity. Regularity then closes when the combined gain exceeds the loss,
\begin{equation}
  \tfrac{\alpha}{2} + \tfrac{\beta}{2} > 1
  \;\Longleftrightarrow\;
  \alpha + \beta > 2 ,
  \label{eq:budget}
\end{equation}
pointing to the line $\alpha+\beta = 2$ rather than the corner $\{\alpha<q\}\cap\{\beta<q\}$.

There is, however, a third reading that the numerics will favor. The current density~\eqref{eq:current} has an out-of-plane part $-\Delta a$ that carries {two} derivatives of $a$, against one derivative of $b$ in the in-plane part $\nperp b$; the scaling~\eqref{eq:J-scaling} loads both equally, but the higher derivative count makes $\Delta a$ the more singular component. Crucially, $\Delta a$ is regularized only by the $a$-equation dissipation $\Lal$. If blow-up concentrates in the out-of-plane current, that is, in a current sheet in $\Delta a$, then the binding constraint is $\alpha$ alone, and the onset boundary should be {near-vertical} in the $(\alpha,\beta)$ plane, controlled chiefly by the dissipation on the magnetic potential. These three readings (corner, line, and $\alpha$-controlled) make distinct, falsifiable predictions that our experiments are designed to separate.

The condition $\alpha+\beta>2$ is sufficient for local well-posedness \citep{Peng2026}; for sufficiently weak dissipation we expect small scales to form and the out-of-plane current to concentrate. We test, in particular, whether the onset of this growth is governed by the symmetric line $\alpha+\beta=2$ or instead primarily by the single exponent $\alpha$, as the out-of-plane current argument suggests, and we seek the self-similar exponent $q^*$ of the observed concentration. We emphasize that local well-posedness for $\alpha+\beta>2$ does not preclude small-scale growth or later blow-up there, so $\alpha+\beta=2$ should be read as the threshold of the present {theory}, not as a proven global regularity boundary.

\section{Numerical method}\label{sec:method}

We integrate~\eqref{eq:emhd} using a Fourier pseudospectral discretization in space and exponential time differencing in time. This treats the linear fractional dissipation exactly at the modal level while advancing the Hall nonlinearity explicitly. This section describes the scheme and its validation.

\subsection{Spatial discretization}

We discretize~\eqref{eq:emhd} by a Fourier pseudospectral method on $\T^2 = [0,2\pi)^2$ with $N^2$ collocation points and spectral resolution $|k_x|,|k_y|\le N/2$. Linear operators act diagonally in Fourier space: the fractional dissipations are the multipliers $|k|^\alpha$ and $|k|^\beta$, and all derivatives are evaluated spectrally. The Hall nonlinearities $a_y b_x - a_x b_y$ and $a_y\Delta a_x - a_x\Delta a_y$ are formed as physical-space products of spectrally computed derivatives and transformed back, with aliasing errors removed by the $2/3$ rule~\citep{Orszag1971}. In the most marginal runs we replace the sharp $2/3$ cutoff by a smooth high-order exponential filter and verify that reported diagnostics are insensitive to the choice.

\subsection{Time integration}

Because the fractional dissipation is linear, diagonal, and stiff, we integrate it exactly through an integrating factor and advance the Hall nonlinearity explicitly with the fourth-order exponential time-differencing Runge--Kutta scheme ETDRK4~\citep{CoxMatthews2002,KassamTrefethen2005}. Writing the system as $\partial_t \hat w = \mathcal{L}\hat w + \mathcal{N}(\hat w)$ with $\mathcal{L} = \mathrm{diag}(-|k|^\alpha,-|k|^\beta)$, the modewise integrating factors are $e^{-|k|^\alpha\Delta t}$ and $e^{-|k|^\beta\Delta t}$; the ETDRK4 coefficient ($\varphi$-)functions are evaluated by the contour-integral quadrature of~\cite{KassamTrefethen2005} to avoid cancellation error at small $|k|$. The time step is chosen from the stability constraint of the explicit treatment of the Hall term, which is {dispersive}: the linearized nonlinearity supports whistler-type waves with frequency $\omega\sim |k|^2\,\|\bB\|$, so explicit stability requires $\Delta t \lesssim c/(k_{\max}^2\,\|\bB\|_\infty)$. This dispersive restriction, rather than the comparatively mild fractional stiffness ($\sim N^{\alpha}$ with $\alpha<2$), is the binding one, and it tightens as small scales form; we choose $\Delta t$ conservatively for each run and stop integrating once the resolution gate of \sect{subsec:resolution} trips.

\subsection{Validation}\label{subsec:validation}

We validate the solver by four independent tests, summarized in Table~\ref{tab:validation} and Figure~\ref{fig:validation}.

\emph{(V1a) Linear dissipation.} A single Fourier mode decays at the exact rate $e^{-|k|^\alpha t}$; the scheme reproduces this to relative error $3.6\times 10^{-16}$, confirming the fractional symbols.

\emph{(V1b) Conservation structure.} For the dissipationless system the instantaneous nonlinear energy rate $\frac{d}{dt}M$ of Proposition~\ref{prop:energy} must vanish. The two contributions from the $a$- and $b$-equations cancel to $2.8\times 10^{-16}$ (with and without dealiasing), verifying the signs and structure of the Hall nonlinearity. With dissipation restored, the discrete energy balance $M(0)-M(T)=\int_0^T D\,dt$, $D=\norm{\Lambda^{1+\alpha/2}a}{\Ltwo}^2+\norm{\Lambda^{\beta/2}b}{\Ltwo}^2$, holds to relative residual $5.5\times 10^{-7}$ (quadrature-limited). We note that the dissipationless system is genuinely ill-posed at the grid scale (the third-derivative Hall term is uncontrolled), so this invariant is verified on \emph{resolved} evolutions; it is the energy balance, not exact conservation, that the solver respects in the dissipative regime of interest.

\emph{(V2) Convergence.} Self-convergence in a smooth, resolved regime confirms fourth-order temporal accuracy (measured rates $4.01$ and $3.94$) and spectral spatial accuracy (error $7.4\times10^{-4}$, $1.8\times10^{-5}$, $1.5\times10^{-7}$ at $N=16,32,64$).

\begin{table}[htbp]
\centering
\caption{Solver validation. Each test compares a measured quantity against the
behavior expected for an ETDRK4 Fourier pseudospectral scheme: the first two
tests probe the exact treatment of the linear and inviscid structure (near
machine precision), the third the discrete energy balance under dissipation, and
the last two the temporal and spatial convergence. All tests pass at the
expected levels, establishing that the diagnostics reported below reflect the
dynamics rather than the discretization.}
\label{tab:validation}
\small
\setlength{\tabcolsep}{4.5pt}
\begin{tabular}{llll}
\toprule
Test & Measured quantity & Expected & Result \\
\midrule
Linear decay (V1a)        & rel.\ error in mode amplitude & mach.\ prec. & $3.6\times10^{-16}$\\
Conservation structure (V1b) & instantaneous $dM/dt$ (inviscid) & mach.\ prec. & $2.8\times10^{-16}$\\
Energy balance (V1b)      & $|M(0){-}M(T)-\!\int\! D\,dt|/|{\Delta M}|$ & quad.\ accuracy & $5.5\times10^{-7}$\\
Temporal order (V2)       & fitted convergence rate & $4$ (ETDRK4) & $4.01,\ 3.94$\\
Spatial accuracy (V2)     & max error, $N=16{\to}64$ & spectral & $7.4{\times}10^{-4}\ \text{to}\ 1.5{\times}10^{-7}$\\
\bottomrule
\end{tabular}
\end{table}

\begin{figure}[htbp]
  \centering
  \includegraphics[width=0.86\linewidth]{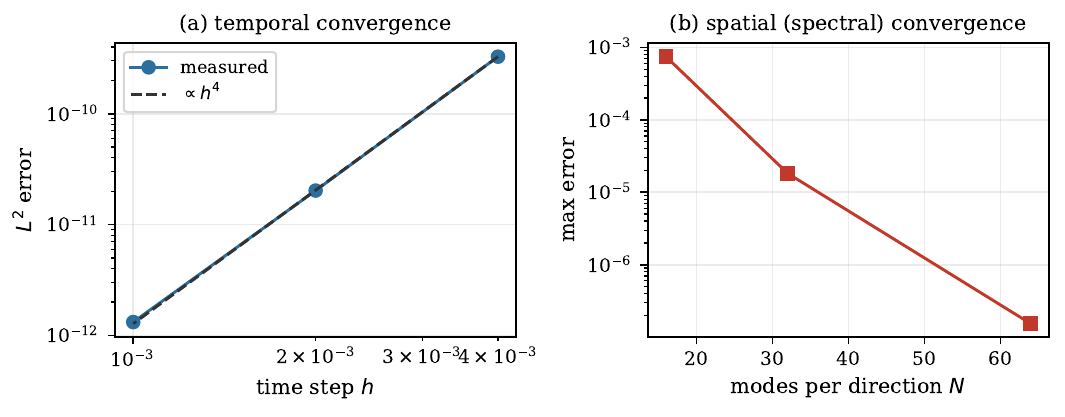}
  \caption{Solver convergence. (a) Fourth-order temporal accuracy of ETDRK4. (b) Spectral spatial accuracy. Together with the machine-precision conservation tests these establish the solver as a trustworthy instrument before any singularity claim.}
  \label{fig:validation}
\end{figure}

\section{Singularity diagnostics and resolution control}\label{sec:diagnostics}

The central difficulty of any numerical study of near-singular dynamics is to distinguish genuine small-scale concentration from under-resolution. We address this with a hierarchy of diagnostics required to agree on a common putative singular time $t^*$, together with a resolution-control protocol.

\subsection{Diagnostics}\label{subsec:diagnostics}

\begin{enumerate}[leftmargin=1.6em,itemsep=3pt]
\item \textbf{Analyticity-strip width $\delta(t)$.} For an analytic field the energy spectrum decays exponentially, $\log E(k,t) \sim C(t) - 2\delta(t)\,k$ at high wavenumber, where $\delta(t)$ is the distance from the real axis to the nearest complex-space singularity~\citep{SSF1983}. A real singularity at $t^*$ is signalled by $\delta(t)\downarrow 0$, typically as $\delta(t)\sim c\,(t^*-t)^\nu$. We extract $\delta(t)$ by least-squares fitting $\log E(k,t)$ over the band in which the spectrum lies between $10^{-3}$ and $10^{-10}$ of its peak (above round-off, below the energy-containing scales).

\item \textbf{Current-density growth.} We track $\norm{\bJ}{\Linf}(t)$ and its in-plane and out-of-plane parts $\norm{\nabla b}{\Linf}$, $\norm{\Delta a}{\Linf}$, fitting an algebraic ansatz $\norm{\bJ}{\Linf}\sim(t^*-t)^{-\gamma_J}$. In practice we monitor the componentwise proxy of~\eqref{eq:Jinf}, which lower-bounds the pointwise supremum of $|\bJ|$ by at most a factor $\sqrt2$; in the runs reported the two differ by at most $11\%$ and exhibit the same growth behavior.

\item \textbf{Sobolev energies.} We monitor $E_s(t) = \norm{a}{H^{s+1}}^2 + \norm{b}{\Hs}^2$ with $s=2$, the numerical analogue of the asymmetric energy in which the local theory of~\cite{Peng2026} closes, both to connect the experiments to the analysis and to detect loss of regularity.

\item \textbf{{Time-integrated current maximum}.} {We monitor $\int_0^t \norm{\bJ}{\Linf}\,d\tau$ as an auxiliary growth diagnostic; it is motivated by criteria of Beale--Kato--Majda type~\citep{BKM1984} but is not an established continuation criterion for the present model.} Its growth provides an additional consistency check for the fitted $t^*$.

\item \textbf{Dissipation channels and the $H^s$ budget.} The local theory controls the nonlinearity by the two dissipation norms
\begin{equation}
  D_a(t)={\norm{\Lambda^{\alpha/2}a}{H^{s+1}}^2},\qquad
  D_b(t)=\norm{\Lambda^{\beta/2}b}{\Hs}^2 .
  \label{eq:channels}
\end{equation}
We monitor these separately, together with the nonlinear inputs to $\tfrac12\frac{d}{dt}E_s$, $\mathcal N_a=-\langle a,\,\nabla^\perp b\cdot\nabla a\rangle_{H^{s+1}}$ and $\mathcal N_b=-\langle b,\,\nabla^\perp a\cdot\nabla\Delta a\rangle_{\Hs}$, and the budget ratio $(\mathcal N_a+\mathcal N_b)/(D_a+D_b)$. The ratio is signed: the numerator is negative when the nonlinearity transfers energy out of the $H^s$ class (in the runs reported this occurs only at the first sample, and negligibly). A ratio persistently above one means the nonlinearity outruns the dissipative capacity in the energy class of the local theory. Throughout, the Sobolev norms are implemented with the inhomogeneous Fourier weights $(1+|k|^2)^{s}$; the energy identity $\frac{d}{dt}E_s = 2(\mathcal N_a+\mathcal N_b)-2\big(\norm{\Lambda^{\alpha/2} a}{H^{s+1}}^2+\norm{\Lambda^{\beta/2}b}{\Hs}^2\big)$ is exact relative to these weights (and holds up to the equivalence of Sobolev weights otherwise), and we verified the implementation against it to a relative $5\times10^{-7}$ by centered differencing along the discrete flow.

\item \textbf{Dyadic shell energies.} Since the local estimates are proved by Littlewood--Paley decomposition, with $\alpha+\beta>2$ entering as a positive summability margin in the dyadic sums, we track the weighted shell energies
\begin{equation}
  E_q(t)=\lambda_q^{2s}\big(\norm{\Lambda\Delta_q a}{L^2}^2
        +\norm{\Delta_q b}{L^2}^2\big),\qquad \lambda_q=2^q,
  \label{eq:shells}
\end{equation}
where the sharp annular projections $2^q\le|k|<2^{q+1}$ stand in for smooth Littlewood--Paley blocks. The profile of $E_q$ across $q$ monitors the dyadic summability underlying the $H^{s+1}\times\Hs$ estimates.

\item \textbf{Cancellation defect.} With the weights $(1+|k|^2)$ on $a$ and $1$ on $b$, the two nonlinear inputs cancel exactly, $\mathcal N_a+\mathcal N_b=0$ at $s=0$: this combines the magnetic-energy cancellation of \sect{sec:structure} with incompressibility of the drift. We monitor the relative defect $r_0(t)=|\mathcal N_a+\mathcal N_b|/(|\mathcal N_a|+|\mathcal N_b|)$ of this identity along the discrete flow. We emphasize that this tests the nonlinear magnetic-energy cancellation at the base energy level, not the full higher-order cancellation of low--high interactions in the Littlewood--Paley proof. It is a structural check that the discretization preserves the base-level cancellation on which the local theory is built, and it doubles as a resolution indicator: the defect sits at round-off while the solution is spectrally resolved and departs from it only when the truncation shell becomes populated.
\end{enumerate}

\subsection{Resolution-control protocol}\label{subsec:resolution}

The diagnostics of \sect{subsec:diagnostics} can each be mimicked by under-resolution: spectral energy piling up at the grid scale can produce spurious current growth and an apparent collapse of the analyticity strip, both artifacts of truncation rather than features of the continuous dynamics. The purpose of this subsection is to make the distinction operational, by fixing in advance the conditions that a reported onset must meet. A current-sheet onset is reported only if it survives the following controls.

\begin{itemize}[leftmargin=1.6em,itemsep=3pt]
\item \emph{Resolution gate.} The solution is trusted only while the analyticity strip remains resolved, $\delta(t)\gtrsim C\,\Delta x$ (several grid points per strip width), equivalently while the spectrum has decayed to round-off before $k_{\max}$. Any extrapolated singular time $t^*$ is inferred from the \emph{resolved} window, never by integrating into the under-resolved regime.

\item \emph{Resolution ladder.} Representative runs are repeated at increasing resolutions with $\Delta t$ reduced accordingly; the trusted-window diagnostics must converge under refinement, and the peak intensity must not decrease with $N$. Only the resolved interval is reported.

\item \emph{Aliasing insensitivity.} Results hold under both the $2/3$ rule and a high-order exponential filter.

\item \emph{Cross-diagnostic consistency.} The signals $\delta(t)\to0$, $\norm{\bJ}{\Linf}\to\infty$, and $E_s(t)\to\infty$ are required to indicate the \emph{same} $t^*$, with exponents consistent with the scaling relations of \sect{sec:scaling}.
\end{itemize}

\subsection{Dynamic rescaling}\label{subsec:rescaling}

To characterize the self-similar structure of a candidate collapse one may use the McLaughlin--Papanicolaou--Sulem--Sulem dynamic-rescaling strategy~\citep{MPSS1986}, normalizing by the current. We outline it here, as it motivates the static exponent estimates of \sect{sec:results}. Introducing a renormalized time $\tau$ through $d\tau = L(\tau)^{-q}\,dt$ and rescaled fields adapted to~\eqref{eq:scaling}, with $L(\tau)$ pinned by the condition $\norm{\bJ}{\Linf}\equiv 1$ in rescaled variables, the exponent $q=q(\tau)$ is determined adaptively and would converge to a value $q^*$ if a self-similar profile exists; a finite-time singularity then corresponds to $t^* = \int_0^\infty L(\tau)^{q}\,d\tau < \infty$. The reduced equations in rescaled variables are recorded in Appendix~\ref{app:rescaling}.

\section{Numerical results}\label{sec:results}

We now report the numerical experiments. All runs follow the protocol of \sect{sec:diagnostics}, and all quantitative claims refer to the trusted window defined there.

\subsection{Validation and numerical setup}\label{subsec:setup}

The solver passes all tests of \sect{subsec:validation} (Table~\ref{tab:validation}, Figure~\ref{fig:validation}).

For reproducibility we record the precise setup used in the remainder of this section. The initial data are the smooth, band-limited shear-type fields
\begin{equation}
\begin{aligned}
  a_0(x,y) &= \cos x + \tfrac12\cos 2y + \tfrac3{10}\sin(x+y), \\
  b_0(x,y) &= \sin y + \tfrac12\sin 2x + \tfrac3{10}\cos(x-y),
\end{aligned}
\label{eq:ic}
\end{equation}
on $\T^2=[0,2\pi)^2$. Each field is made mean-free and then rescaled by a common factor, so that $\max(\|a_0\|_\infty,\|b_0\|_\infty)=\epsilon$. The amplitude $\epsilon$ and the pair $(\alpha,\beta)$ are the only quantities varied, leaving the dissipation as the sole control parameter. Aliasing is removed by the $2/3$ rule (with the smooth exponential-filter variant used as a cross-check in the most marginal runs). The time step is fixed per run at the dispersive-CFL value of \sect{sec:method},
\begin{equation}
  \Delta t = \frac{c_{\rm CFL}}{k_{\max}^2\,\max(\epsilon,1)},
  \qquad k_{\max}=N/3,\quad c_{\rm CFL}=0.18 .
\end{equation}
Because the whistler frequency involves the solution amplitude, which is not constant along a run, we verified the adequacy of this fixed step directly: {repeating the singular reference run at $N=128$ with $\Delta t$ halved changes $\norm{\bJ}{\Linf}(t)$ by at most $0.6\%$, and its peak by $0.03\%$, over $0\le t\le0.6$, which contains the whole integration window including the growth phase; an adaptive step is therefore not required for the diagnostics reported here.} Diagnostics are sampled every $20$ steps; a run is integrated until the resolution gate trips, i.e.\ until the analyticity strip falls to the grid scale, $\delta(t)<\Delta x$, after which the solution is no longer trusted. {This gate is a stopping criterion, not a guarantee of full resolution: in the sense of \sect{subsec:resolution} the solution is well resolved only while the strip spans several grid cells, and values attained close to the gate, such as the peak currents of the resolution ladder in \sect{subsec:ladder}, are used only to observe trends under refinement.}

A run is classified \emph{singular} when the current grows by more than a factor $1.6$ while the strip collapses below $2\,\Delta x$, \emph{regular} when the current never grows by more than a factor $1.25$, and \emph{marginal} otherwise. Current growth is the primary classifier; the strip condition reinforces the singular label rather than defining the regular one, because for weakly dissipated runs the saturated strip can settle below $2\,\Delta x$ without collapsing further, while the current shows no growth at all. In our data the two criteria never conflict: every run with growth above $1.6$ also shows strip collapse toward the gate. {The cut-offs are conventions; their influence on the classification and on the onset boundary is quantified in \sect{subsec:phase}.}

\subsection{The \texorpdfstring{$(\alpha,\beta)$}{(alpha,beta)} phase diagram: an asymmetric, \texorpdfstring{$\alpha$}{alpha}-controlled boundary}
\label{subsec:phase}

Fixing the initial data and sweeping $(\alpha,\beta)$, we classify each run as \emph{singular} (the out-of-plane current grows and the analyticity strip collapses to the grid scale), \emph{regular} (the current decays and the strip saturates), or \emph{marginal}, using the resolution-controlled diagnostics of \sect{sec:diagnostics}. Figure~\ref{fig:phase}(a) shows the outcome {on an $8\times8$ grid of $(\alpha,\beta)\in[0.4,1.8]^2$} at a representative amplitude. The singular region occupies small $\alpha$ and the regular region large $\alpha$, separated by a near-vertical onset band; the symmetric line $\alpha+\beta=2$ cuts \emph{across} the boundary rather than tracing it.

The asymmetry is sharpest along our sampled points on the line $\alpha+\beta=2$ (Table~\ref{tab:line}): holding the sum fixed and increasing $\alpha$, the outcome runs from singular through marginal to regular. Two runs with identical sum, $(\alpha,\beta)=(0.4,1.6)$ and $(1.6,0.4)$, fall on opposite sides of the onset band. To localize the transition we sampled four additional points on the line; along $\alpha+\beta=2$ the transition occurs within the sampled interval $0.85\le\alpha\le1.15$, centered on the marginal point $\alpha=1$. The sum does not appear to determine the onset; the dissipation $\alpha$ on the magnetic potential does. The growth itself is monotone along the line, decreasing from $3.0$ at $\alpha=0.4$ to no growth at all for $\alpha\ge1.3$ (Table~\ref{tab:line}). This matches the third reading of \sect{sec:scaling}: because the concentration occurs in the out-of-plane current $\Delta a$, which only $\Lal$ regularizes, $\alpha$ is the candidate binding parameter.

The $\alpha$-controlled asymmetry is not tied to the shear-type family~\eqref{eq:ic}. {As a robustness check along the line $\alpha+\beta=2$ (not a recomputation of the full phase diagram), we classified four initial-data families at $\epsilon=3$, $N=96$, each made mean-free and rescaled so that $\max(\norm{a_0}{\Linf},\norm{b_0}{\Linf})=\epsilon$:
\begin{itemize}[leftmargin=1.6em,itemsep=1pt]
\item \emph{shear}: the fields~\eqref{eq:ic};
\item \emph{random-phase}: independent standard Gaussian Fourier coefficients (real and imaginary parts) on $1\le|k|\le3$, with $a_0,b_0$ the real parts of the inverse transforms (NumPy \texttt{default\_rng(11)});
\item \emph{higher-wavenumber}: $a_0=\cos(x+y)+0.5\sin(2x-y)+0.3\cos3y$, $b_0=\sin(x-y)+0.5\cos(x+2y)+0.3\sin3x$;
\item \emph{anisotropic}: $a_0=\cos x+0.5\sin2x+0.3\cos x\cos y$, $b_0=\cos2y+0.5\sin y+0.3\sin x\sin y$.
\end{itemize}
In every family the current growth decreases with $\alpha$ along the line (Figure~\ref{fig:icrobust}), so the $\alpha$-ordering is robust. The location of the onset depends on the data: it lies near $\alpha\approx1$ for the first three families and near $\alpha\approx1.3$--$1.6$ for the anisotropic data.}

\begin{table}[htbp]
\centering
\caption{Classification and current growth across our sampled points on the line $\alpha+\beta=2$ at fixed data ($N=96$, $\epsilon=3$). The same sum yields opposite outcomes, and the transition occurs within the sampled interval $0.85\le\alpha\le1.15$, indicating that the observed onset boundary is $\alpha$-controlled rather than symmetric.}
\label{tab:line}
\small
\begin{tabular}{cccccccccc}
\toprule
$\alpha$ & $0.4$ & $0.55$ & $0.7$ & $0.85$ & $1.0$ & $1.15$ & $1.3$ & $1.45$ & $1.6$\\
$\beta$  & $1.6$ & $1.45$ & $1.3$ & $1.15$ & $1.0$ & $0.85$ & $0.7$ & $0.55$ & $0.4$\\
\midrule
outcome & sing. & sing. & sing. & sing. & marg. & reg. & reg. & reg. & reg.\\
$\max_t\norm{\bJ}{\Linf}/\norm{\bJ}{\Linf}(0)$
        & $3.0$ & $2.6$ & $2.2$ & $1.9$ & $1.5$ & $1.1$ & $1.0$ & $1.0$ & $1.0$\\
\bottomrule
\end{tabular}
\end{table}

\begin{figure}[htbp]
  \centering
  \includegraphics[width=\linewidth]{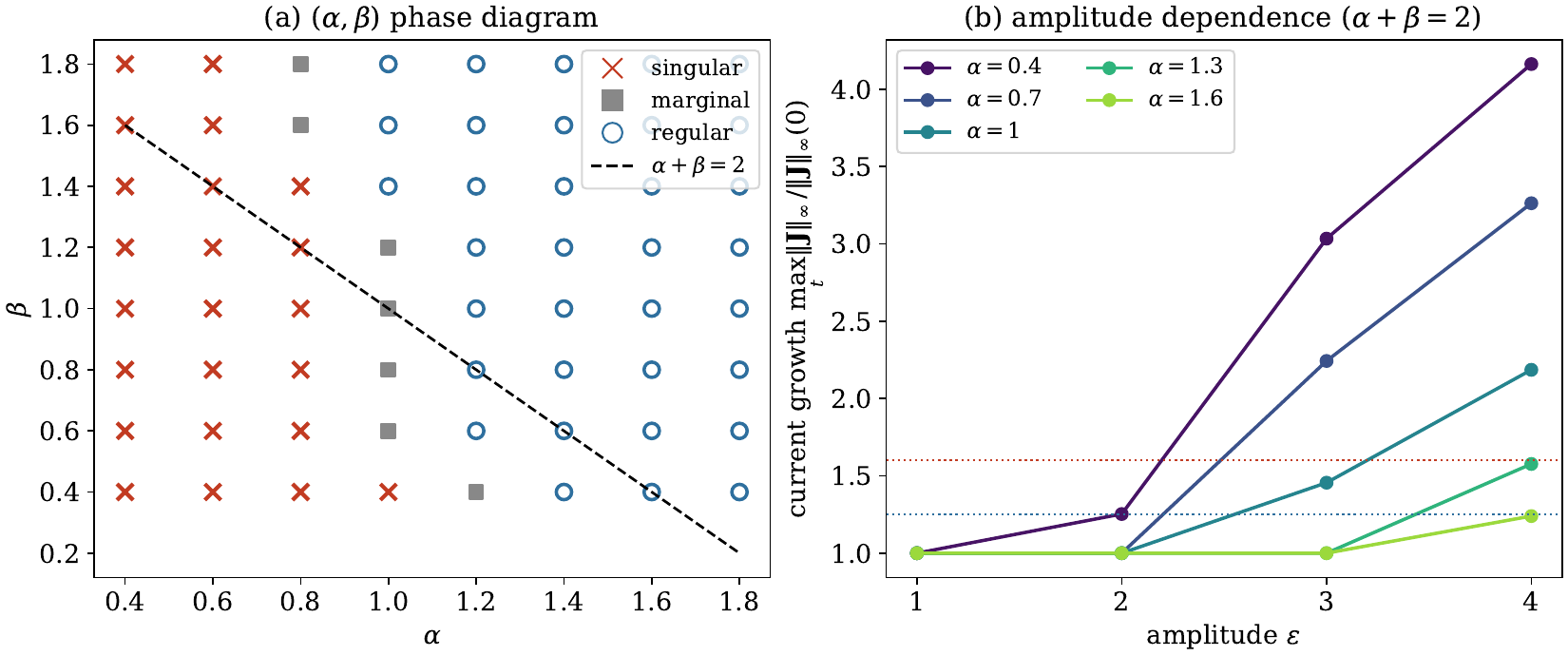}
  \caption{(a) Phase diagram in the $(\alpha,\beta)$ plane at fixed initial data {($8\times8$ grid, $\epsilon=3$, $N=96$)}. Singular points (crosses) occupy small $\alpha$, regular points (circles) large $\alpha$, with a near-vertical onset band; the line $\alpha+\beta=2$ (dashed) does not separate them. (b) The onset is amplitude dependent: current growth versus data amplitude $\epsilon$ {along $\alpha+\beta=2$ for five values of $\alpha$ (dotted lines: classification cut-offs)}. A case regular at small amplitude becomes singular as the data grows.}
  \label{fig:phase}
\end{figure}

\begin{figure}[htbp]
  \centering
  \includegraphics[width=0.66\linewidth]{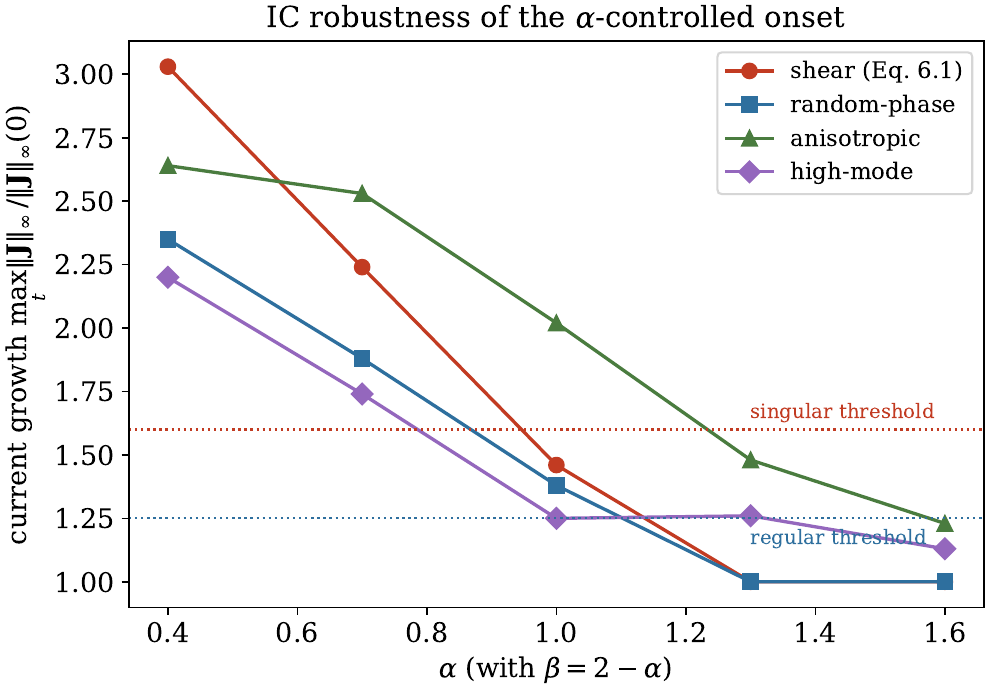}
  \caption{{Current growth along $\alpha+\beta=2$ ($N=96$, $\epsilon=3$) for four initial-data families. Growth decreases with $\alpha$ in every family; the onset location depends on the data. Dotted lines: classification cut-offs.}}
  \label{fig:icrobust}
\end{figure}

The onset boundary is amplitude dependent (Figure~\ref{fig:phase}(b)): a case that is regular at small amplitude becomes singular at larger amplitude, and the boundary advances toward larger exponents as the data grows {(for example, $\alpha=1.0$ shows no growth for $\epsilon\le2$, is marginal at $\epsilon=3$ and singular at $\epsilon=4$)}. The amplitude-dependent growth is not in conflict with the local well-posedness result for $\alpha+\beta>2$, which does not assert global regularity or uniform suppression of small-scale amplification{; the observed boundary is a data-dependent, finite-time observation rather than a regularity threshold}.

{To assess the role of the classification cut-offs we reclassified all runs with the singular/regular cut-offs shifted from $(1.5,\,1.15)$ to $(1.8,\,1.4)$. Over this range at most $5$ of the $64$ phase-diagram points and $2$ of the $20$ amplitude-sweep points change class; for every $\beta$ the onset boundary moves by at most one grid step ($\Delta\alpha=0.2$), and its near-vertical shape is unchanged. The location of the onset band is thus uncertain at the level of the $\alpha$-sampling, while its orientation is insensitive to the cut-offs.}

\subsection{The singular mechanism and a resolution ladder}\label{subsec:ladder}

Figure~\ref{fig:dichotomy} contrasts {a run with $\alpha+\beta<2$ ($\alpha=\beta=0.8$) with a run with $\alpha+\beta>2$ ($\alpha=\beta=1.3$)} at the same data, the fields~\eqref{eq:ic} at amplitude $\epsilon=3$; the same pair of runs is used throughout this subsection and the next two. In the singular case the analyticity strip collapses from $\mathcal{O}(1)$ to the grid scale (panel a) while the current maximum grows (panel b); in the regular case the strip saturates well above the grid scale and the current decays monotonically. Panel (c) identifies the mechanism: the out-of-plane current $\norm{\Delta a}{\Linf}$ dominates the in-plane current $\norm{\nabla b}{\Linf}$ throughout, so the incipient singularity is a sheet in $\Delta a$.

Refining the grid sharpens rather than softens the concentration. {We reran the singular reference case on the resolution ladder $N=96,128,192,256,512$ (Figure~\ref{fig:ladder}). The peak current before the resolution gate rises monotonically, $12.5\to15.4\to19.4\to24.1\to37.9$, roughly as $N^{0.66}$ and without saturation, and at every resolution the turnover of the current coincides with the strip reaching the grid scale. These observations support a truncation-related interpretation of the rollover seen at fixed $N$ in Figure~\ref{fig:dichotomy}(b). The peaks are attained close to the grid scale and are not grid-converged values; at $N=512$ the sampled current becomes oscillatory near the gate, and we quote a three-sample envelope.} While under-resolution can in principle generate spurious oscillations, the increase of the peak current under refinement, taken together with the consistent analyticity-strip collapse and the spectral broadening of Figure~\ref{fig:selfsim}(c), is consistent with a {genuine} concentration mechanism rather than a discretization artifact. Beyond the point at which the strip reaches the grid scale the run is no longer trusted; the apparent turnover of the current in Figure~\ref{fig:dichotomy}(b) {occurs near this resolution limit}. Figure~\ref{fig:field} shows the characteristic sheet structure of $-\Delta a$ near the end of the trusted window.

\begin{figure}[htbp]
  \centering
  \includegraphics[width=\linewidth]{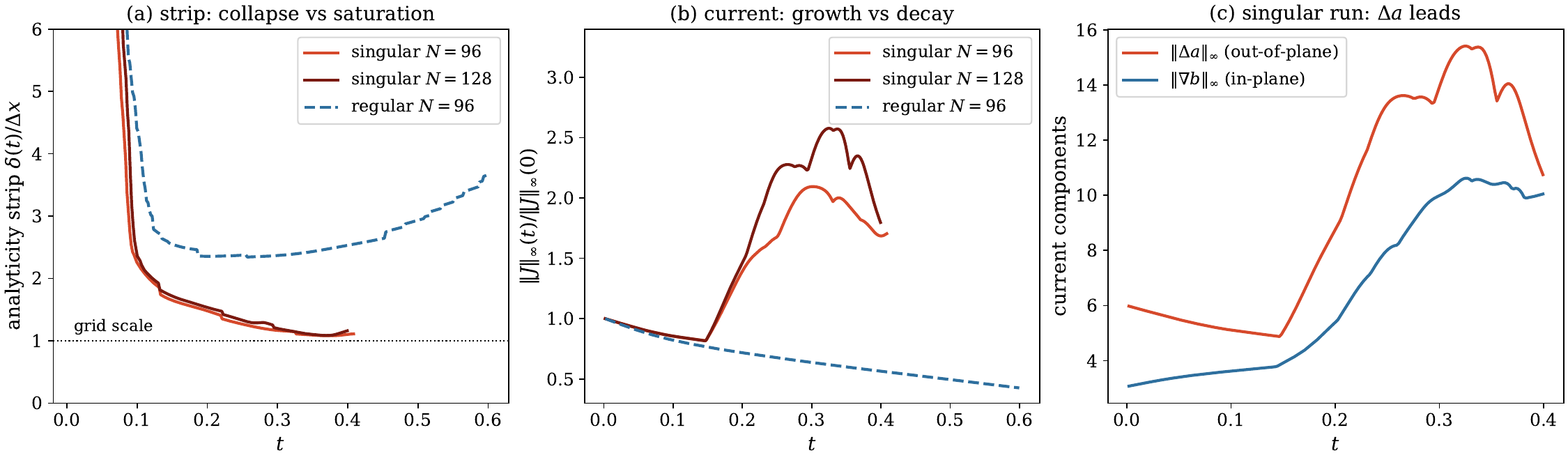}
  \caption{Singular versus regular dynamics. (a) Analyticity-strip width $\delta(t)/\Delta x$: collapse to the grid scale (singular, two resolutions) versus saturation (regular). (b) Current maximum, normalized: growth that increases with resolution (singular) versus monotone decay (regular); the turnover of the singular curves marks the resolution limit. (c) In the singular run the out-of-plane current $\Delta a$ leads the in-plane current $\nabla b$, identifying a current-sheet collapse.}
  \label{fig:dichotomy}
\end{figure}

\begin{figure}[htbp]
  \centering
  \includegraphics[width=0.45\linewidth]{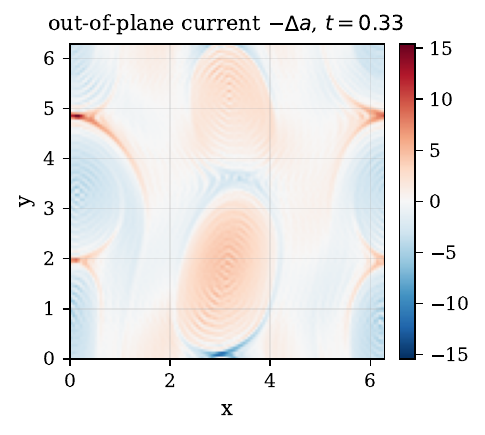}
  \caption{Out-of-plane current $-\Delta a$ near the end of the trusted window of the singular run, showing concentration into a sheet.}
  \label{fig:field}
\end{figure}

\begin{figure}[htbp]
  \centering
  \includegraphics[width=\linewidth]{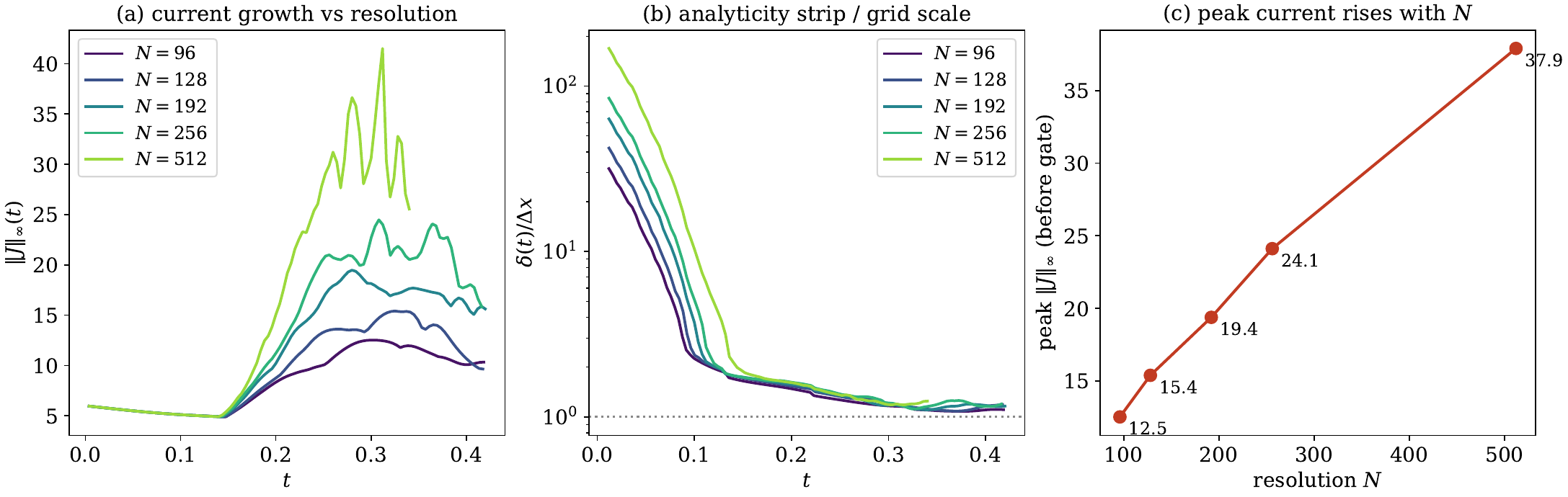}
  \caption{{Resolution ladder for the singular reference case ($\alpha=\beta=0.8$, $\epsilon=3$), $N=96$--$512$. (a) Current maximum; (b) analyticity-strip width in grid units; (c) peak current before the gate versus $N$. The turnover moves to larger values with $N$ and coincides with $\delta$ reaching the grid scale, supporting a truncation-related interpretation of the fixed-$N$ rollover of Figure~\ref{fig:dichotomy}(b).}}
  \label{fig:ladder}
\end{figure}

\subsection{Energy-budget and dyadic-shell diagnostics}
\label{subsec:crossdiag}

The remaining diagnostics of \sect{subsec:diagnostics} confirm the picture and tie it to the energy mechanism of the local theory (Figures~\ref{fig:esbkm} and~\ref{fig:mechanism}); the same quantities admit a direct physical reading as the natural description of anisotropic small-scale energy transfer between the two magnetic components. The energy $E_2(t) = \norm{a}{H^{3}}^2 + \norm{b}{H^{2}}^2$, the numerical analogue of the asymmetric energy in which the local theory of~\cite{Peng2026} closes, grows by a factor of about $8$ at $N=96$ and about $14$ at $N=128$ over the trusted window of the singular run, with the growth accelerating as the strip approaches the grid scale and \emph{increasing} under refinement; in the regular run it decays monotonically, by a factor of about $7$ over the longer integration. The {time-integrated current maximum} $\int_0^t\norm{\bJ}{\Linf}\,d\tau$ is convex in the singular run and concave in the regular one, consistent with accelerating current growth in the former and relaxation in the latter. All signals of the protocol of \sect{subsec:resolution} (strip collapse, current growth, Sobolev-energy growth, convexity of the {time-integrated current maximum}) thus turn on together over the same window of the singular run, and all of them relax in the regular run.

\begin{figure}[htbp]
  \centering
  \includegraphics[width=0.92\linewidth]{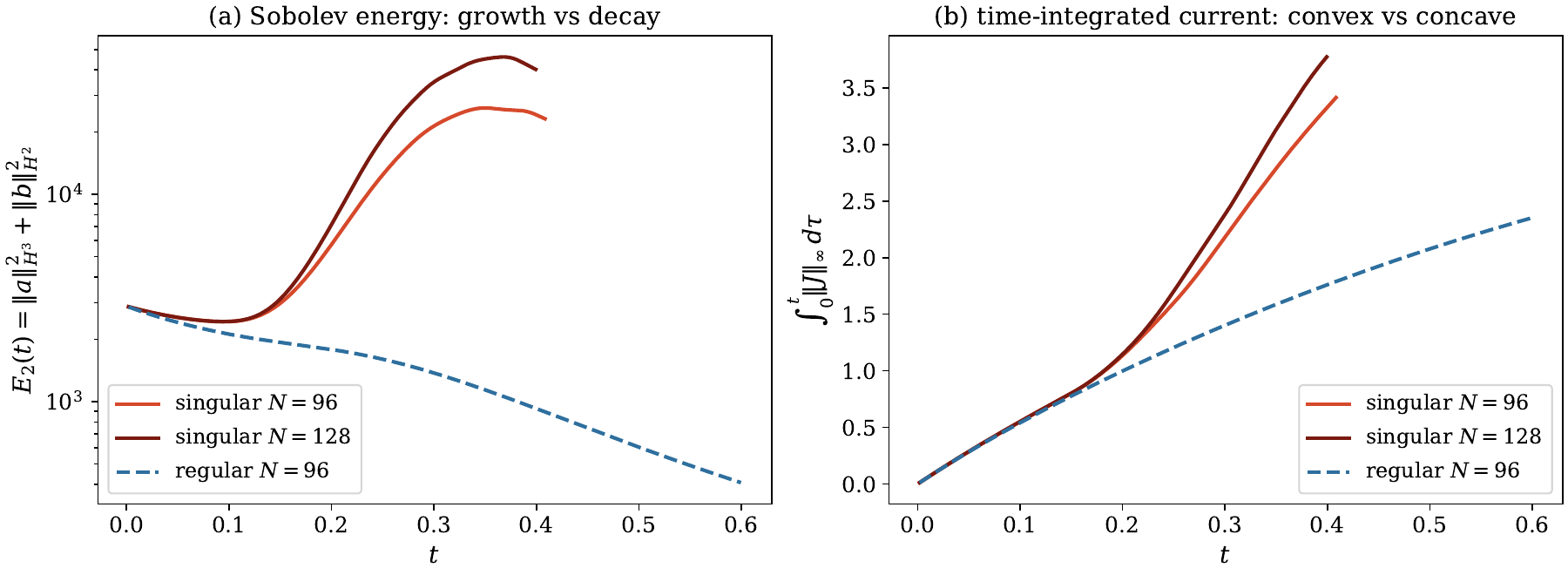}
  \caption{Cross-diagnostic consistency for the runs of Figure~\ref{fig:dichotomy}. (a) The well-posedness Sobolev energy $E_2(t)$ grows in the singular run, faster at higher resolution, and decays monotonically in the regular run. (b) The {time-integrated current maximum} $\int_0^t\norm{\bJ}{\Linf}\,d\tau$ is convex (singular) versus concave (regular).}
  \label{fig:esbkm}
\end{figure}

Figure~\ref{fig:mechanism} resolves this energy balance into the ingredients of the local estimates. Panel (a) shows the two dissipation channels~\eqref{eq:channels}. In the regular run both channels decay after an initial transient, and the $H^2$ budget ratio of panel (c) stays below one throughout (maximum {$0.83$}): the dissipation absorbs the nonlinear input, as the local estimates require. In the singular run the budget crosses one at {$t\approx0.10$}, before the current growth onset at $t\approx0.14$, and reaches {$2.15$}: from that point the nonlinearity outruns the combined dissipative capacity in the energy class, and $E_2$ grows accordingly. The $a$-channel carries the larger dissipative load initially ({$D_a/D_b\approx4$}); as the sheet forms, the two channels become comparable, so once the concentration is underway neither channel alone suffices. This does not contradict the $\alpha$-controlled onset: the concentration is initiated in the $\Delta a$ current channel, which only $\Lal$ regularizes, while the subsequent proof-level energy balance involves both dissipative channels.

Panel (b) shows the dyadic shell energies~\eqref{eq:shells} of the singular run. The weighted profile inverts: by the end of the trusted window $E_q$ \emph{increases} with $q$, with $E_5/E_1$ rising from $2\times10^{-3}$ at $t=0.1$ to about $50$, so the dyadic sums behind the $H^{s+1}\times\Hs$ estimates lose summability from the top shells. This is consistent with the dyadic high-frequency failure mode against which the condition $\alpha+\beta>2$ supplies the summability margin in the local theory. In the regular run (not shown) the weighted profile remains steeply decreasing, $E_5/E_1\le2\times10^{-2}$ throughout, and the top shells peak near $t\approx0.26$ before decaying by one to two orders of magnitude.

Finally, the cancellation defect $r_0(t)$ confirms that the discretization preserves the inter-equation structure: in the $N=128$ singular run it stays below $10^{-15}$ while the strip exceeds $2\,\Delta x$ and rises only to about $10^{-12}$ at the resolution gate. The defect departs from round-off precisely when the truncation shell becomes populated, so it independently supports the trusted window.

\begin{figure}[htbp]
  \centering
  \includegraphics[width=\linewidth]{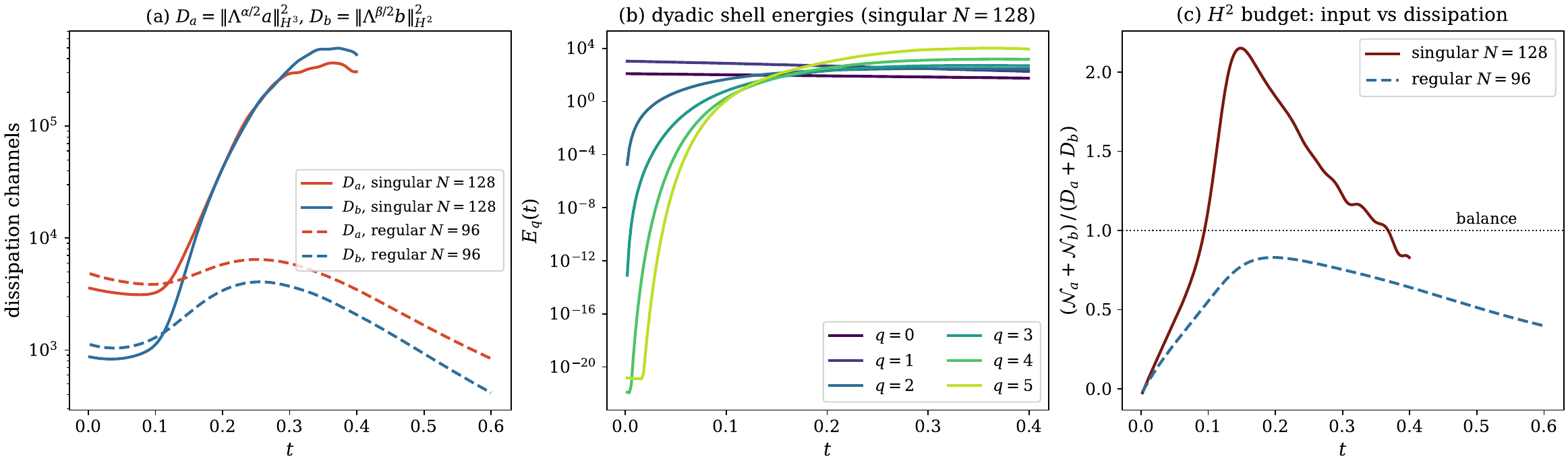}
  \caption{Proof-level energetics ($s=2$) for the runs of Figure~\ref{fig:dichotomy}: singular run $\alpha=\beta=0.8$ at $N=128$ (solid), regular run $\alpha=\beta=1.3$ at $N=96$ (dashed), both at $\epsilon=3$. (a) Dissipation channels $D_a$, $D_b$ of the local theory: both decay in the regular run and grow in the singular run. (b) Dyadic shell energies $E_q$ of the singular run: the weighted profile inverts, with the highest shells dominant by the end of the trusted window. (c) Signed ratio of nonlinear input to dissipation in the $H^2$ energy balance; values above one mean the nonlinearity outruns the combined dissipation. The ratio stays below one in the regular run and crosses one at {$t\approx0.10$} in the singular run, before the onset of current growth at $t\approx0.14$.}
  \label{fig:mechanism}
\end{figure}

\subsection{Self-similar concentration and \texorpdfstring{{a tentative}}{a tentative} exponent
\texorpdfstring{$q^*\approx3$}{q*=3}}

On the growth window of the singular run the current and the strip follow power laws in $t^*-t$ (Figure~\ref{fig:selfsim}a,b), where $t^*$ is an extrapolated putative singular time. The fit window is necessarily narrow: it is bounded below by the resolution gate (its lower end abuts strip widths of one to two grid cells) and above by the onset of growth, so $t^*-t$ spans only a factor of about $1.5$. Within it, the individual exponents $\gamma_J$ (current) and $\nu$ (strip) are only loosely determined and trade off against $t^*$: the measured values $\gamma_J=2.39$ and $\nu=1.16$ at $N=128$ exceed the self-similar predictions $\gamma_J=(q^*-1)/q^*\approx0.67$ and $\nu=1/q^*\approx0.33$ of~\eqref{eq:exponent-relation} by a common factor of about $3.5$. This discrepancy is consistent with sensitivity to the extrapolated value of $t^*$ over a narrow fitting window (see the discussion after~\eqref{eq:exponent-relation}), rather than necessarily indicating a failure of the self-similar scaling. Their ratio, through \begin{equation} 
q^* = 1 + \gamma_J/\nu , 
\end{equation} 
{gives} $q^*=3.03$ at $N=96$ and $q^*=3.06$ at $N=128$. {Because this ratio is formed from two individually poorly determined exponents, we regard $q^*\approx3$ as a tentative estimate rather than a measured value.} The value $q^*\approx 3$ is distinguished: it is the energy-critical scaling at which the conserved magnetic energy is invariant, $M\sim\lambda^{6-2q}=\lambda^{0}$ (Proposition~\ref{prop:scaling} and~\eqref{eq:J-scaling}). {If confirmed, this would mean that the} concentration proceeds at the energy-preserving self-similar rate: the magnetic energy neither concentrates nor disperses under the rescaling, while the current $\norm{\bJ}{\Linf}\sim\lambda^{1-q^*}\sim\lambda^{-2}$ grows. The spectra of Figure~\ref{fig:selfsim}(c) are consistent with the picture: as time advances the exponential tail extends and shallows, \ie\ the analyticity strip closes.

\begin{figure}[htbp]
  \centering
  \includegraphics[width=\linewidth]{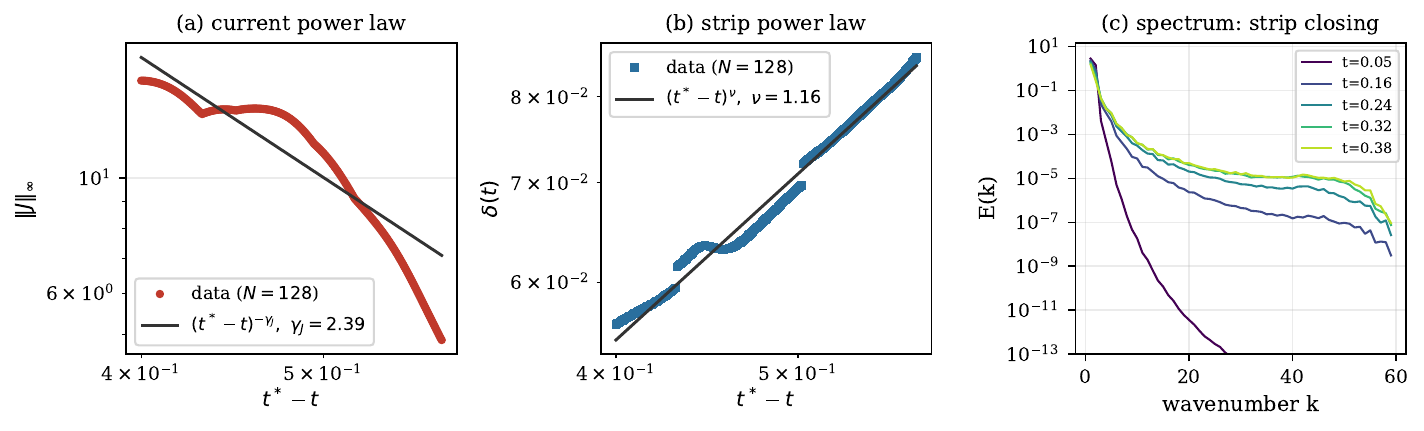}
  \caption{Self-similar structure of the singular run ($\alpha=\beta=0.8$, $\epsilon=3$, $N=128$). (a) Current maximum and (b) analyticity-strip width versus $t^*-t$, with fitted power laws; the ratio {gives a tentative estimate} $q^*\approx 3$. (c) Energy spectra at successive times: the exponential tail extends and shallows as the analyticity strip closes.}
  \label{fig:selfsim}
\end{figure}

{Appendix~\ref{app:rescaling} formulates the dynamic-rescaling method, but we have not integrated the rescaled system: an adaptive exponent $q(\tau)$ introduces additional terms in the rescaled equations, and a meaningful determination of $q^*$ requires following the collapse substantially closer to $t^*$ than the resolutions available here allow. The static fits above therefore give only a preliminary estimate, and a full dynamic-rescaling computation, together with a tight separation of $\gamma_J$ and $\nu$, is left for future work.}

\section{Conclusion} \label{sec:conclusion}

In this paper, we have studied the role of split fractional dissipation in the $2.5$D EMHD system by combining scaling analysis with numerical experiments under controlled grid refinement. Taking the local well-posedness condition $\alpha+\beta>2$ as a reference balance, we compared three candidate pictures: a split-dependent corner, the symmetric line $\alpha+\beta=2$, and an $\alpha$-controlled boundary tied to the out-of-plane current. The computations favor the third picture: across sampled points on the line $\alpha+\beta=2$, the behavior changes from concentration at small $\alpha$ to decay at large $\alpha$. Since the sum $\alpha+\beta$ is fixed along this line, this transition cannot be attributed to the combined dissipation alone; it indicates instead that the dissipation on the magnetic potential is the controlling parameter. The concentrating runs form an out-of-plane current sheet in $\Delta a$, sharpen under grid refinement, and show consistent collapse of the analyticity strip, growth of the asymmetric Sobolev energy, inversion of weighted dyadic shell energies, and nonlinear transfer exceeding the combined dissipative budget. {Power-law fits on the growth window suggest, tentatively, an energy-critical self-similar rate $q^*\approx 3$.} These findings provide resolution-controlled numerical evidence for a concentration mechanism, but not a proof of finite-time singularity{; nor do they contradict or refine the local well-posedness theory for $\alpha+\beta>2$, which guarantees a solution on a data-dependent time interval and makes no statement about transient small-scale growth}.

The observed asymmetry also suggests a practical implication for EMHD and Hall--MHD computation. Within this reduced split-dissipation model, distributing the same nominal damping budget differently between the two magnetic components can lead to qualitatively different small-scale behavior, with damping of the magnetic potential playing the decisive role in gating current-sheet formation. {The asymmetric, dissipation-dependent behavior suggests several questions for analysis.
\begin{enumerate}[leftmargin=1.8em,itemsep=1pt]
\item \emph{Component-resolved estimates.} A natural goal is a priori estimates, or a continuation criterion, in which control of the out-of-plane current $\Delta a$, and hence the dissipation $\Lal$, is the binding requirement, while $b$ enters only through the transport field $\nperp b$. Such a result would explain the dominance of $\Delta a$ observed here while retaining the coupling through $\beta$.
\item \emph{Transient amplification.} Local well-posedness for $\alpha+\beta>2$ does not bound transient growth. Estimates of the amplification of $\norm{\bJ}{\Linf}$ over a fixed time in terms of the data amplitude, the data geometry and the two exponents would make the amplitude- and geometry-dependence of the onset quantitative.
\item \emph{Singularity versus instability.} Where no local theory is available ($\alpha+\beta\le2$), an analytical criterion distinguishing finite-time concentration from high-frequency instability is needed, together with existence or exclusion conditions for self-similar profiles of the rescaled equations of Appendix~\ref{app:rescaling}.
\end{enumerate}}

\section*{Acknowledgements}
This work was partially supported by the ONR grant \#N00014-24-1-2432, the Simons Foundation (MP-TSM-00002783), and the NSF grant DMS-2420988.

\bibliographystyle{plainnat}
\bibliography{references}

\FloatBarrier
\appendix
\section{Rescaled equations}\label{app:rescaling}

Following \sect{subsec:rescaling}, introduce the rescaled space and time variables $\xi = (x-x^*)/L(\tau)$ and $d\tau = L(\tau)^{-q}\,dt$, and the rescaled fields consistent with the scaling~\eqref{eq:scaling},
\begin{equation}
  a(x,t) = L^{\,3-q}\,A(\xi,\tau), \qquad b(x,t) = L^{\,2-q}\,B(\xi,\tau).
\end{equation}
Writing $\kappa := (\ln L)_\tau$ for the renormalized growth rate, a direct computation transforms the dissipationless system~\eqref{eq:inviscid} into the autonomous form
\begin{subequations}\label{eq:rescaled}
\begin{align}
  \partial_\tau A + \nperp B\cdot\nabla A
    &= \kappa\big[\,\xi\cdot\nabla A - (3-q)A\,\big], \\
  \partial_\tau B + \nperp A\cdot\nabla(\Delta A)
    &= \kappa\big[\,\xi\cdot\nabla B - (2-q)B\,\big],
\end{align}
\end{subequations}
where all spatial operators are now in $\xi$. The scale $L(\tau)$ is fixed by the normalization that the rescaled current have unit amplitude, $\norm{\bJ_\xi}{\Linf}\equiv 1$, which by~\eqref{eq:J-scaling}, where $\bJ_{\rm phys}\sim L^{1-q}\bJ_\xi$, amounts to $L = \norm{\bJ}{\Linf}^{1/(1-q)} = \norm{\bJ}{\Linf}^{-1/(q-1)}$ (so that, since $q>1$, the scale $L\to 0$ as the current intensifies); {these equations are written for a fixed exponent $q$. In an adaptive implementation, in the manner of~\citep{MPSS1986}, $q=q(\tau)$ is chosen to keep the rescaled fields stationary in a chosen norm, and the change of variables then produces the additional terms $q_\tau(\ln L)\,A$ and $q_\tau(\ln L)\,B$ on the right-hand sides of~\eqref{eq:rescaled}.} A self-similar collapse corresponds to a $\tau$-independent profile $(A,B)$ with {constant $q=q^*$} and constant $\kappa$; the finite singular time is $t^* = \int_0^\infty L(\tau)^{q}\,d\tau<\infty$.

\end{document}